\documentclass[authoryear]{elsarticle}
\usepackage{amsthm,amsmath,amssymb}
\usepackage{natbib,url,enumerate}
\usepackage[margin=0pt]{subfig}
\usepackage{tikz}
\usetikzlibrary{positioning,arrows,decorations.pathreplacing}
\usepackage{hyperref}
\usepackage{xcolor}
\hypersetup{
    colorlinks,
    linkcolor={red!50!black},
    citecolor={blue!50!black},
    urlcolor={blue!80!black}
}

\newtheorem{thm}{Theorem}
\newtheorem{prop}[thm]{Proposition} 
\newtheorem{lemma}[thm]{Lemma}

\theoremstyle{definition}
\newtheorem{defi}{Definition}
\theoremstyle{remark}

\newcommand*{\set}[1]{\left\{#1\right\}} 

\newcommand*{\good}{G} 
\newcommand*{\bad}{B} 

\begin{document}
\begin{frontmatter}
\title{Good signals gone bad: dynamic signalling with switching efforts} 
\author{Sander Heinsalu\fnref{thanks}}
\address{Research School of Economics, Australian National University, HW Arndt Building 25a Kingsley St, Acton, ACT 2601, Australia.
\newline Email: sander.heinsalu@anu.edu.au, 
website: \url{http://sanderheinsalu.com/}} 
\fntext[thanks]{
The author is greatly indebted to Johannes H\"{o}rner for many enlightening discussions about this research. Numerous conversations with Larry Samuelson have helped improve the paper. The author is grateful to Eduardo Faingold, Dirk Bergemann, Juuso V\"{a}lim\"{a}ki, Philipp Strack, Willemien Kets, Tadashi Sekiguchi, Flavio Toxvaerd, Fran\c{c}oise Forges, Jack Stecher, Daniel Barron, Benjamin Golub, Florian Ederer, Vijay Krishna, Sambuddha Ghosh, Sergiu Hart, Idione Meneghel, J\"{o}rgen Weibull, Alessandro Bonatti, Zvika Neeman, Scott Kominers, Andrzej Skrzypacz, George Mailath, Joel Sobel, anonymous referees and participants of many conferences and seminars for comments and suggestions. Any remaining errors are the author's. Financial support from Yale University and the Cowles Foundation is gratefully acknowledged.
}

\begin{abstract}
This paper examines signalling when the sender exerts effort and receives benefits over time. Receivers only observe a noisy public signal about the effort, which has no intrinsic value. 

The modelling of signalling in a dynamic context gives rise to novel equilibrium outcomes. In some equilibria, a sender with a higher cost of effort exerts strictly more effort than his low-cost counterpart. 
The low-cost type can compensate later for initial low effort, but this is not worthwhile for a high-cost type. The interpretation of a given signal switches endogenously over time, depending on which type the receivers expect to send it. 

\end{abstract}
\end{frontmatter}

\section{Introduction}

In many signalling situations, the sender exerts effort over time, and the observation of that effort is noisy. For example, a politician may be a (relatively) honest or a corrupt type, and can signal honesty by following the law to the letter (paying taxes in full, refraining from speeding and bribe-taking). Whereas the politician incurs the cost of obeying the law at all points of time, voters learn of low effort only after the realisation of a random event, such as a scandal. The honest type has a lower cost of obeying the law. The honest type acts in the public interest in important matters, but the corrupt type does not. The voters care about decisions in important matters, but not directly about whether the politician obeys the law in everyday life. 

As a consequence of the multiple opportunities of exerting effort, novel dynamics of behaviour arise.
For example, there exist equilibria in which the high-cost type chooses a strictly higher effort level initially than the low-cost type. 

In the model, the players are a sender and a competitive market of receivers. The sender is either a high-cost or a low-cost type. The type is private information. Receivers share a common prior belief about the type. The sender continuously chooses his effort level. Receivers observe noisy public signals about the effort, rather than the effort itself. The signal process is Poisson with intensity decreasing in effort. The types only differ in their flow cost of effort. The sender derives a flow benefit directly from the posterior belief of the receivers. 

Attention is restricted to equilibria that are Markovian (the belief of the receivers is the state variable) and stationary, which means that behaviour does not depend on calendar time. Such equilibria exist, because minimal effort by all types in all circumstances is a Markovian and stationary equilibrium and always exists. 

In some parameter regions, there exist equilibria in which first one type exerts higher effort and then the other. These are called \emph{switched effort equilibria}, because initially the ordering of the efforts of the types is the opposite to that found in the previous literature on signalling. The concept of a switched effort equilibrium can be illustrated by the example of the politician given above. In this example, the politician can be honest or corrupt and can
exert effort to obey the law. Lawful behaviour decreases the frequency of scandals. A switched effort equilibrium can be described in terms of four regimes, which are referred to as \emph{early career}, \emph{insider}, \emph{scrutiny} and \emph{tainted}. Play starts in the early career, during which the corrupt type exerts positive effort and the honest type no effort. If no scandal occurs by a given time, then the politician becomes an insider, which means that the voters ignore scandals and the politician no longer exerts effort. If, instead, a scandal occurs in the early career, then scrutiny results. Under scrutiny, the honest type exerts maximum effort and the corrupt type none. Under scrutiny, a scandal leads to a tainted reputation, which means that voters believe the politician to be corrupt with higher probability than in the other regimes, and the politician exerts no effort. The path of play in switched effort equilibria is depicted in Fig.~\ref{fig:regimes} below. 

\begin{figure}[h]
	\centering
	\begin{tikzpicture}
	\node (init) {Early career};
	\node[right=5em of init] (in) {Insider}
	edge[>=stealth,<-] node[above, font=\small] {no scandal} (init.east);
	\node[below right=5em of init] (scru) {Scrutiny}
	edge[>=stealth,<-]  node[below, rotate=-45,font=\small] {scandal} (init.east)
	edge[>=stealth,->,bend right=90] node[above,font=\small] {no scandal} (scru.north);
	\node[right=5em of scru] (ta) {Tainted}
	edge[>=stealth,<-] node[below,font=\small] {scandal} (scru.east);
	\end{tikzpicture}
	\caption{A switched effort equilibrium.}
	\label{fig:regimes}
\end{figure}
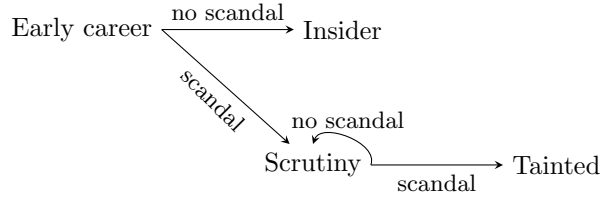

In the early career regime of switched effort equilibria, a signal has the opposite meaning to that in the other regimes.\footnote{One example of a scandal increasing popularity is Bill Clinton, whose approval rating rose from 63\% on 15 Dec 1998 to 73\% on 20 Dec. The House of Representatives passed two articles of impeachment in that period. 
Toronto mayor Rob Ford's approval rating rose from 39\% on 28 Oct 2013 to 44\% on 1 Nov. Between these dates, police chief Bill Blair confirmed a video of Ford smoking crack cocaine.} 
If the low-cost type is more likely to generate a particular signal, then this signal is evidence of low cost, whereas if the high-cost type becomes more likely to send it, then the signal suggests high cost to the receivers. 

The voters do not care directly about the part of the politician's behaviour that leads to scandals (sex life, drug use, tax evasion), but care about the politician's unobservable decisions on issues important to the country. The voters assume that a politician's tendency for scandalous behaviour in matters that have no direct impact on the public is positively correlated with a tendency of corrupt decisions in areas where these cause real damage. These combined tendencies are called a corrupt type. The voters cannot observe the decisions that affect them, so they use scandals to update their belief about the likelihood of honesty in such decisions. 

The politician derives a flow benefit from the voters' belief, which can be justified by a randomly arriving election. This arrival is exogenous to the politician if he is not the head of government and cannot call a snap election. Because of the chance of an election, belief at each future instant has a positive probability of mattering for the politician's career. Alternatively, the politician may derive ego rents from being popular with the public.

\subsection{Literature}

Signalling is used to explain phenomena as diverse as education \citep{spence1973}, conspicuous consumption \citep{veblen1899} and issuing equity \citep{leland+pyle1977}. Many authors have mentioned the relevance of time \citep{weiss1983,admati+perry1987} and noise \citep{matthews+mirman1983} in signalling contexts. The conclusion of the previous literature is that in all equilibria, the high-cost type exerts weakly less effort than the low-cost. 
This result fails to hold only in models that depart from pure Spence signalling, e.g.\ by adding exogenous information revelation and more than two types as in \cite{feltovich+2002}. Introducing noise \citep{matthews+mirman1983} or repetition \citep{noldeke+damme1990,swinkels1999} has not been shown to switch the efforts of the types in pure signalling. Similarly, such switching of efforts has not been found to arise in models that incorporate both dynamics and imperfect monitoring \citep{daley+green2012b,gryglewicz2009,dilme2014}. 

In \cite{dilme2014}, signalling in continuous time is obscured by Brownian noise. The sender (an entrepreneur) decides how much costly effort to exert over time, as well as when to stop the game (sell the firm) and receive a final benefit. This contrasts with the present paper, in which benefit accumulates continuously and the sender cannot stop the receivers from learning. \cite{dilme2014} studies the efficiency of equilibrium effort provision and finds that effort is too low and stops too early. The present paper examines the endogenous switching of the efforts of the types, which is not alluded to in the previous literature, including \cite{dilme2014}. 

In \cite{daley+green2012b}, the uninformed traders receive information (observations of a diffusion process) exogenously over time, and the informed trader decides when to stop the game (execute the trade) and receive a final payoff. \cite{gryglewicz2009} examines limit pricing over time. The low-cost incumbent is a commitment type and the high-cost incumbent decides when to stop imitating the low-cost type. Unlike \cite{daley+green2012b}, the present paper models an endogenous signal. Moreover, the present paper differs from Gryglewicz (2009) in that both types are strategic. The situation described is one in which payoffs accrue continuously or a lump-sum arrives at an exogenous exponentially distributed time. In contrast, the other papers model a stopping decision by the sender that triggers a payoff. Neither \cite{daley+green2012b} nor \cite{gryglewicz2009} refers to switched efforts or interpretations of a signal. 

The present paper is also related---albeit less closely---to the literature on repeated
noiseless signalling \citep{noldeke+damme1990,swinkels1999}. In these models, time is discrete, there is no noise, and the sender pays the signalling cost first and receives the benefit only upon deciding to stop signalling forever. For example, the completion of a traditional education can be modelled in this way---the salary is received only after graduating. 
In the current paper, the benefit is received concurrently with the payment of the cost, as when a worker takes continuing education courses while being employed, or a firm advertises while selling its product. 
\cite{noldeke+damme1990} find a unique informative equilibrium. Using different informational assumptions, \cite{swinkels1999} finds a unique pooling equilibrium. The models in the current paper have many informative equilibria and one pooling equilibrium. In \cite{noldeke+damme1990} and \cite{swinkels1999}, the low-cost type always exerts weakly more effort than the high-cost. While this feature of their models is in line with previous signalling research, it stands in contrast with the current work. 

Discrete-time repeated signalling is also studied in \cite{kaya2009} and \cite{roddie2012b}. In their models, there is no noise in the observation of the sender's action, whereas in the present paper this observation is noisy. \cite{kaya2009} focuses on least-cost separating equilibria. \cite{roddie2012b} provides general conditions for reputation effects to arise. 
Neither paper mentions the possibility of switched efforts, which are the focus of the current paper. 

It is, in fact, easier to construct switched effort equilibria without noise, on account of the availability of belief threats. The author is not aware, however, of any papers in which this has been done. Using belief threats or ad hoc refinements invites the suspicion that switched efforts are driven by unrealistic beliefs. This concern is addressed in the current work by using noise, so that Bayes' rule applies after all histories. Similarly, non-Markov strategies can be used to create strange behaviour, but this paper stacks the deck against unusual results by focussing on Markov stationary equilibria. This serves to strengthen the result that effort switching is possible. 

Switched effort equilibria could not be studied with existing noisy signalling models, which use either discrete time or Brownian noise. Discrete time with noise makes equilibrium calculation intractable. Brownian noise is incompatible with switched efforts\footnote{The proof is available upon request. It uses a coupling argument on the continuous belief paths.} and is less tractable than Poisson for the models in this paper. 

Switched efforts are reminiscent of the countersignalling of \cite{feltovich+2002}, but the mechanism driving them is quite different, as is the model. The switched effort equilibria of this paper use dynamics, while \cite{feltovich+2002} present a one shot model. In the current paper, there is no exogenous information revelation and there are only two types, in contrast to \cite{feltovich+2002}. In countersignalling, the lowest-cost type can rely on the exogenous signal to partly distinguish him from the highest-cost. Moreover, he may exert less effort than the medium type in order to differentiate himself from that type. The effort of the lowest-cost type cannot be less than that of the highest-cost type. In the present paper, it is the threat of future information revelation that incentivises the high-cost type to signal. This threat is not as severe for the low-cost type, leading to strictly less effort. 

\cite{cripps+2004} show that, in a wide class of repeated games, if a reputation for behaviour is not an equilibrium of the complete information stage game, then that reputation is temporary, and the type must eventually be learned. In some switched effort equilibria of the present paper, both types face a positive probability of acquiring a `wrong' permanent reputation in the following sense: when signalling and belief updating stop, belief about the good type may be lower than the prior and belief about the bad type higher. In expectation, beliefs move in the direction of the sender's type, but mistakes have positive probability.

\section{Setup}
\label{sec:setup}

The players are a strategic sender and a competitive market. The sender has a type $\theta\in\set{\good,\bad}$, with $\good$ interpreted as the good (low-cost) type and $\bad$ as bad (high-cost). 
The sender knows his type, the market does not. The initial log likelihood ratio $l_0\in\mathbb{R}$ of the types is common knowledge. Throughout this paper, log likelihood ratio $l$ is used instead of belief $\Pr(\good) = \frac{\exp(l)}{1+\exp(l)}$, as this simplifies the formulas in the dynamic models to follow. All results can be restated in terms of beliefs. A generic log likelihood ratio $l$ is an element of $\overline{\mathbb{R}}=\mathbb{R}\cup\set{\infty,-\infty}$. The log likelihood ratio corresponding to $\Pr(\good)=1$ is $l=\infty$ and corresponding to $\Pr(\good)=0$ is $l=-\infty$. 

Time is continuous and the horizon is infinite. The sender chooses signalling effort $e_t\in[0,1]$ at each instant of time $t$. To avoid the technical difficulties of defining behavioural strategies in continuous time, mixing is not allowed. Effort $e$ costs type $\theta$ sender $c_{\theta}(e)$, with $c_{\theta}$ continuously differentiable, strictly increasing, convex and $c_{\theta}(0)=0$.
\emph{Strong single crossing} $c'_{\good}(1)<c'_{\bad}(0)$ is assumed, which implies the single crossing in type and cost that is standard in the signalling literature. 
Strong single crossing means that the marginal cost of any effort level for the good type is lower than the marginal cost of any effort for the bad type. 

Effort benefits the sender via its effect on the signal process, which drives the market's log likelihood ratio process $(l_t)$. The sender is assumed to derive flow benefit $\beta(l)$ directly from the market's log likelihood ratio $l$. 
Unless noted otherwise, the function $\beta$ is assumed to be strictly increasing, bounded and continuously differentiable. Denote the flow benefit from $l=\infty$ (corresponding to $\Pr(\good)=1$) by $\beta_{\max}$ and from $l=-\infty$ by $\beta_{\min}$. 

In order to provide microfoundations that explain why the market rewards the sender's type rather than the effort it expects from the sender, suppose there are two kinds of effort: effort which is noisily observable and that which is unobservable. 
The senders with a low cost of partially observable effort prefer to exert the unobserved effort, but the high-cost senders prefer not to exert it. Only the completely unobservable kind of effort matters to the market. The type is the level of this hidden effort. The partially visible effort is the signalling activity, which the market does not care about directly. 

A \emph{pure Markov stationary strategy} is a measurable function $(e_{\bad},e_{\good}):\overline{\mathbb{R}}\rightarrow [0,1]^2$ that maps the log likelihood ratio to the efforts of the types. The state variable $l$ is the log likelihood ratio of the market. Formally, the value of the state variable at time $t$ is 
the left limit $l_{t-}$ of the log likelihood ratio process $(l_t)$, similarly to \cite{yushkevich1988}. The log likelihood ratio $l_t$ is a function of time. The left limit means that time approaches $t$ from below. This assumption ensures that the signals are not anticipated by the log likelihood ratio of the market. Intuitively, the market does not condition on the signal realisation in the `next instant'. The convention $l_{0-}=l_0$ is used for the initial value of the state variable. 
Henceforth, only pure Markov stationary strategies are considered and the `pure Markov stationary' phrase is omitted. 

The signal process is Poisson with rate $(1-e_t)\lambda+d$ at time $t$.\footnote{A Poisson rate linear in effort is with some loss of generality, but less than at first seems. A change of variables $\hat{e}=f^{-1}(e)$ transforms cost to $c_{\theta}(f(\hat{e}))$ and the signal rate to $(1-f(\hat{e}_t))\lambda+d$, which may be nonlinear. A strictly increasing convex $f$ preserves the convexity and the strong single crossing of cost.} 
The parameter $\lambda\in(0,\infty)$ is interpreted as the informativeness of effort, while $d\geq 0$ is the minimal rate. 
The meaning of signals is endogenous---whether a signal raises or lowers belief about the type depends on the strategy the market expects from the sender. The same signal can thus have opposite meanings at different log likelihood ratios. 
The market observes the signals, but not the sender's effort or type. Since the signal process is public, the updated log likelihood ratio is common knowledge.

The next part of the setup is the Bayesian updating of the market's log likelihood ratio. 
Denote the strategy the market expects the sender to choose by $e^*=(e_{\good}^*,e_{\bad}^*)$. This notation is also used for equilibrium strategies. If a signal occurs at log likelihood ratio $l\in\mathbb{R}$, then
the log likelihood ratio jumps to 
\begin{align}
\label{jumpb}
j(l)=
l+\ln\left(\frac{\lambda(1-e_{\good}^*(l))+d}{\lambda(1-e_{\bad}^*(l))+d}\right).
\end{align}
Call $|j(l)-l|$ the \emph{jump length}. 
If $d=0$, then the fraction in~(\ref{jumpb}) may become $\frac{0}{0}$ for some strategy $e^*$ that the market expects. In that case, updating uses the limit of~(\ref{jumpb}) as $d\searrow0$. This implies that $\frac{0}{0}=1$, and if $l=\pm\infty$, then $l$ stays constant regardless of signals. 
Note that updating is defined for any $e^*$, not just the equilibrium $e^*$ introduced below. The set of strategies available to the sender is independent of $d$. The limit $d\searrow0$ is taken for each $e^*$ separately to define the overall updating rule that maps $e^*$ and the signal history to $l$. 
The use of the $d\searrow0$ limit of Bayes' rule applies Bayes' rule to null events. 

Belief updating based on Bayes' rule is summarised in the following lemma. 
\begin{lemma}
\label{def:poissonloglr}
Let $(e_{\good}^*,e_{\bad}^*)$ be the strategy the market expects. Fix $t\geq0$. Suppose signals have occurred at $\tau_1,\ldots,\tau_n$, with $0\leq \tau_1\leq\ldots\leq \tau_n\leq t$. Then the log likelihood ratio at $t$ is
\begin{align}
\label{poissonbnewsloglr}
l_t=l_0+\lambda\int_0^t\left[e_{\good}^*(l_{s})-e_{\bad}^*(l_{s})\right]ds+\sum_{k=1}^n \left[j(l_{\tau_k})-l_{\tau_k}\right].
\end{align}
\end{lemma}
Except for jumps, $l$ evolves deterministically given the market expectations $(e_{\bad}^*,e_{\good}^*)$. Given $l$ at the time of a jump, the jump length is deterministic. The log likelihood process depends on the chosen effort only via the (random) timing of the signals $\tau_1,\ldots,\tau_n$. 
Lemma~\ref{def:poissonloglr} applies even if multiple signals occur in the same instant $\tau_i$, but this event has zero probability, because the signal process is Poisson. 

If for some $\hat{l}$, the strategy features $e_{\good}^*(l)>e_{\bad}^*(l)$ for $l\in (\hat{l}-\epsilon,\hat{l})$ and $e_{\good}^*(l)<e_{\bad}^*(l)$ for $l\in (\hat{l},\hat{l}+\epsilon)$, then $\hat{l}$ is a \emph{stasis point}. A stasis point is described in more detail in the appendix. It occurs at a point $\hat{l}$ that has the log likelihood ratio drift towards $\hat{l}$ from immediately above and below $\hat{l}$. The log likelihood ratio does not drift away from a stasis point, but may jump away in either direction. 

Given the strategy $e^*=(e_{\bad}^*,e_{\good}^*)$ that the market expects the sender to choose, the payoff of type $\theta$ from actually choosing the effort function $e_{\theta }(\cdot)$ 
is the expected discounted sum of flow payoffs
\begin{align}
\label{Poissonvalue}
J^{e_{\theta}}_{l_0}(e^*)=\mathbb{E}^{e_{\theta}}\left[\int_0^{\infty}\exp(-rt)\left[\beta(l_t)-c_{\theta}e_{\theta }(l_t)\right]dt\middle|e^*,l_{t=0}=l_0\right],
\end{align}
where the expectation is over the stochastic process $(l_t)_{t\geq0}$, given $e_{\theta}(\cdot)$. Payoffs both on and off the equilibrium path are given by~(\ref{Poissonvalue}), depending on whether or not the $e_{\good},e_{\bad}$ that maximise~(\ref{Poissonvalue}) for each type satisfy $(e_{\good},e_{\bad})=(e^*_{\good},e^*_{\bad})$. 
The discount rate is $r>0$. 

Given a strategy $e^*$ that the market expects from the sender, the supremum of~(\ref{Poissonvalue}) over $e_{\theta}$ is denoted $V_{\theta}(l_0)$. 
If the market expects a Markov stationary strategy, then every time a given $l$ is reached, the continuation value $V_{\theta}(l)$ of type $\theta$ is well defined and independent of the path of $(l_t)$ that led to $l$. The dependence of $V_{\theta}(l)$ on $e^*$ is suppressed in the notation. If $e_{\good}^*(l)=e_{\bad}^*(l)$, then $(l_t)$ stays at $l$ forever and $V_{\theta}(l)=\int_0^{\infty}\exp(-rt)\beta(l)dt=\frac{\beta(l)}{r}$. 
Value $V_{\theta}(l)$ is bounded above by $\frac{\beta_{\max}}{r}$ and below by $\frac{\beta_{\min}}{r}$. 

\begin{defi}
\label{def:Poissonequil}
A \emph{Markov stationary equilibrium} consists of a strategy $e^*=(e_{\good}^*, e_{\bad}^*)$ of the sender and a log likelihood ratio process $(l_t)_{t\geq0}$ s.t.\
\begin{enumerate}
\item given $(l_t)_{t\geq0}$, $e_{\theta}^*$ maximises~(\ref{Poissonvalue}) over $e_{\theta}$,
\item given $e^*$, $(l_t)_{t\geq0}$ is derived from~(\ref{poissonbnewsloglr}). 
\end{enumerate}
\end{defi}
Henceforth `equilibrium' means a pure Markov stationary equilibrium. 

A \emph{pooling} equilibrium is defined by $e_{\bad}^*(l)=e_{\good}^*(l)\;\forall l$. It is clearly Markovian (independent of past play conditional on $l$) and stationary (independent of calendar time). A pooling equilibrium exists for all parameter values, because if $e_{\bad}^*(l)=e_{\good}^*(l)$, then the log likelihood ratio stays constant forever at $l$. If a log likelihood ratio is unresponsive to effort, then there is no benefit to signalling. This implies that there is no incentive to exert effort at $l$. The unique best response is $e_{\bad}(l)=e_{\good}(l)=0$. Existence is thus guaranteed for the equilibrium concept in Def.~\ref{def:Poissonequil}. 


The focus in this paper is on equilibria in which $e_{\bad}^*(l)>e_{\good}^*(l)=0$ for a nonempty open set of $l$. In the introduction, such equilibria were called \emph{switched effort equilibria}. The introduction also described the path of play of such equilibria, which is depicted in Fig.~\ref{fig:regimes}.  

A switched effort equilibrium cannot be unique, because the pooling equilibrium always exists. It will be shown below that if one switched equilibrium exists, then a continuum of such equilibria exist. 
Refinements cannot be used to select an equilibrium, because Bayes' rule applies everywhere, as explained above. 


\section{Switched effort equilibria}

For some parameter values, there exist equilibria in which, for some log likelihood ratios of the market, the $\bad$ type exerts higher effort than $\good$, despite the uniformly higher marginal cost of effort. The result is reminiscent of the countersignalling of \cite{feltovich+2002}, but the mechanism is quite different. In this model, it is the threat of future information revelation that incentivises $\bad$ to signal. This threat is not as severe for $\good$.

The switched effort pattern $e_{\bad}^*(l)>e_{\good}^*(l)=0$ is counterintuitive, because the flow benefit from a higher log likelihood ratio is the same for the types, but $\bad$ has a higher marginal cost of signalling. The strong single crossing makes the result more stark. The switched effort pattern requires $V_{\bad}$ to decrease at some $l$, because if $e_{\bad}^*(l)>e_{\good}^*(l)$, then $j(l)>l$ and $\bad$ is paying a cost to avoid jumps. For this to be optimal, $V_{\bad}(j(l))<V_{\bad}(l)$ is necessary. 
The decrease of $V_{\bad}$ in $l$ means that the discounted payoff of the $\bad$ type decreases in the probability that the receivers assign to him being $\good$. 

A switched effort equilibrium is constructed as follows: first, assume that the appropriate strategy is employed; second, calculate the value functions; finally, check that the strategy is a best response at every $l$. 
A sketch of a switched effort equilibrium is in Fig.~\ref{fig:poissonstrange2}, in which the initial log likelihood ratio $l_0$ is assumed to be in the interval $(\underline{l},l_1)$. To maximise the duration of switched efforts, let $l_0\rightarrow l_1$. 

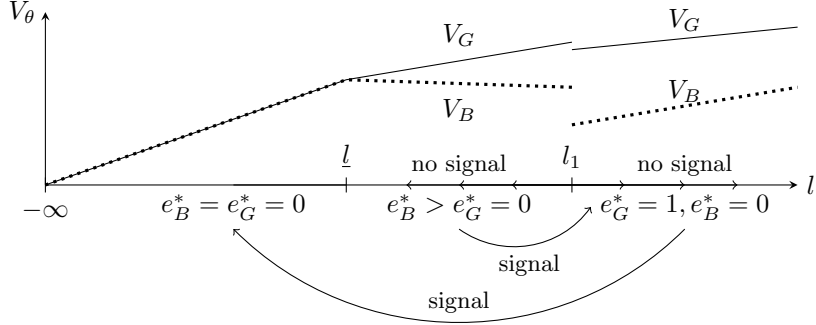
\begin{figure}[h]
	\caption{Value functions and strategy in a switched effort equilibrium. Early career regime: $e_{\bad}^*>e_{\good}^*=0$, $l\in(\underline{l},l_1]$, insider: $e_{\bad}^*=e_{\good}^*=0$, $l=\underline{l}$, scrutiny: $e_{\good}^*=1,e_{\bad}^*=0$, $l> l_1$, tainted: $e_{\bad}^*=e_{\good}^*=0$, $l<\underline{l}$.}
	\label{fig:poissonstrange2}
	\centering
	\begin{tikzpicture}
	\draw[->,>=stealth] (0,0)--(10,0) node (paxis) [right] {$l$};
	\draw[->,>=stealth] (0,0)--(0,2.3) node (vaxis) [left] {$V_{\theta}$};
	\draw (0,0.1)--(0,-0.1) node[below] (zero) {$-\infty$};
	\draw (4,-0.1)--(4,0.1) node[above] {$\underline{l}$};
	\draw (4,0)--(2.5,0) node[below] (elequaleh) {$e_{\bad}^*=e_{\good}^*=0$};
	\draw[->] (7,0)--(5.5,0) node[below] (elgreatereh)  {$e_{\bad}^*>e_{\good}^*=0$}
	node[above,font=\small]  {no signal};
	\draw[->] (7,0)--(6.2,0);
	\draw[->] (7,0)--(4.8,0);
	\draw (7,-0.1)--(7,0.1) node[above] {$l_1$};
	\draw[->] (7,0)--(8.5,0) node[below] (eh1el0) {$e_{\good}^*=1,e_{\bad}^*=0$}
	node[above,font=\small]  {no signal};
	\draw[->] (7,0)--(9.2,0);
	\draw[->] (7,0)--(7.7,0);
	\path[->,bend right=45] (elgreatereh.south) edge node[below,font=\small] {signal} (eh1el0.west);
	\path[->,bend left=45] (eh1el0.south) edge node[above,font=\small] {signal} (elequaleh.south);
	\draw[very thick,dotted] (0,0)--(4,1.4)--(7, 1.3);
	\draw[very thick,dotted] (7, 0.8)--(10, 1.3);
	\node at (5.5,1) {$V_{\bad}$};
	\node at (8.5,1.25) {$V_{\bad}$};
	\draw (0,0)--(4,1.4)--(7, 1.9);
	\draw (7,1.8)--(10, 2.1);
	\node at (5.5,2) {$V_{\good}$};
	\node at (8.5,2.2) {$V_{\good}$};
	\end{tikzpicture}
\end{figure}

Prop.~\ref{prop:poissonstrangeequil2} below provides sufficient conditions for a switched equilibrium to exist. These involve the value functions, which are endogenous, but all conditions can be expressed using primitives, as shown subsequently in Prop.~\ref{prop:poissonstrangeequil3}. 

For those values of $l$ at which $e_{\bad}^*(l)=e_{\good}^*(l)$ (the insider and tainted states in Fig.~\ref{fig:regimes}), incentives are trivial, because $l$ does not respond to signals according to~(\ref{jumpb}). The unique best response, therefore, is $e_{\bad}(l)=e_{\good}(l)=0$. The same reasoning was used above to show the existence of the pooling equilibrium. 

It remains to check the log likelihood ratio regions with $e_{\good}^*(l)=1$, $e_{\bad}^*(l)=0$ (scrutiny in Fig.~\ref{fig:regimes}) and $e_{\bad}^*(l)>e_{\good}^*(l)=0$ (early career in Fig.~\ref{fig:regimes}). There are incentive constraints at each $l$ for both types, but this continuum can be reduced to just five inequalities, depicted in Fig.~\ref{fig:incent}. There is one inequality for each type in the scrutiny region; one inequality for $\good$ in the early career; and
a further two inequalities for $\bad$ in the early career.  

To check that the best responses are $e_{\good}(l)=1,e_{\bad}(l)=0$ when   $e_{\good}^*(l)=1$, $e_{\bad}^*(l)=0$, it is enough to verify two conditions: (i) $\good$ has an incentive to avoid jumps when the avoidance motive (value before minus after jump) is minimal; and (ii) $\bad$ does not have this incentive when the avoidance motive is maximal. The condition for $\good$ to exert maximal effort is Prop.~\ref{prop:poissonstrangeequil2}~(a) below and the condition for $\bad$ to choose zero effort is Prop.~\ref{prop:poissonstrangeequil2}~(b). These are also shown in Fig.~\ref{fig:incent}. 

In the region where $e_{\bad}^*(l)>e_{\good}^*(l)$, it must be checked that $\good$ has no incentive to avoid jumps and $\bad$ has neither too much nor too little incentive, so that the best response of $\bad$ is interior. Define 
\begin{align*}
\overline{V}^{>}_{\theta} &=\sup\set{V_{\theta}(l):e_{\bad}^*(l)>e_{\good}^*(l)},
\\ \underline{V}^{>}_{\theta} &=\inf\set{V_{\theta}(l):e_{\bad}^*(l)>e_{\good}^*(l)},
\\ \overline{V}^1_{\theta} &=\sup\set{V_{\theta}(l):e_{\good}^*(l)=1,e_{\bad}^*(l)=0},
\\ \underline{V}^1_{\theta} &=\inf\set{V_{\theta}(l):e_{\good}^*(l)=1,e_{\bad}^*(l)=0}.
\end{align*}
The greatest temptation for $\good$ to avoid jumps in the early career occurs at the maximal value $\overline{V}^{>}_{\good}$ when jumps go to just above $l_1$ (see Fig.~\ref{fig:incent}). 
If the difference between $\overline{V}^{>}_{\good}$ and $V_{\good}(l_1) =\underline{V}^1_{\good}$ does not incentivise $\good$ to exert effort (condition (c) in Prop.~\ref{prop:poissonstrangeequil2}), then neither do other $V_{\good}(l)$ and $V_{\good}(j(l))$ in these regions. 

The equilibrium conditions that determine the effort of $\bad$ in the early career are $\lambda[V_{\bad}(l)-V_{\bad}(j(l))]=c_{\bad}'(e_{\bad}^*(l))$ and $j(l)=l+\frac{\lambda+d}{\lambda(1-e_{\bad}^*(l))+d}$. The first is the FOC for $\bad$, ensuring that $\bad$ is optimising given the jump length. The second condition determines the jump length given the equilibrium effort $e_{\bad}^*(l)$, using $e_{\good}^*(l)=0$.
A sufficient condition for $\bad$ to choose the equilibrium level of effort when $e_{\bad}^*(l)>e_{\good}^*(l)=0$ is as follows: jumps at $\overline{V}^{>}_{\bad}$ going to $\overline{V}^{1}_{\bad}$ do not incentivise $\bad$ to exert maximal effort (Prop.~\ref{prop:poissonstrangeequil2}~(d)) and jumps at $\underline{V}^{>}_{\bad}$ going to $\underline{V}^{1}_{\bad}$ do motivate a positive level of effort (Prop.~\ref{prop:poissonstrangeequil2}~(e)). Then by the Mean Value Theorem, an interior effort level and jump length can be found such that the jump is derived from the effort using~(\ref{jumpb}) and the effort is a best response to the jump. 

\begin{figure}[h]
	\caption{Incentives in a switched effort equilibrium with $d=0$. Prop.~\ref{prop:poissonstrangeequil2} (a)--(e) are the vertical double-ended arrows.}
	\label{fig:incent}
	\centering
	\begin{tikzpicture}
	\draw[->,>=stealth] (0,0)--(10,0) node (paxis) [right] {$l$};
	\draw[->,>=stealth] (0,0)--(0,2.3) node (vaxis) [left] {$V_{\theta}$};
	\draw (0,0.1)--(0,-0.1) node[below] (zero) {$-\infty$};
	\draw (3,0.1)--(3,-0.1) node[below] {$\underline{l}$};
	\draw (5.5, 0.1)--(5.5, -0.1) node[below] {$l_0$};
	\draw (6,0.1)--(6,-0.1) node[below] {$l_1$};
		\draw (9,0.1)--(9,-0.1) node[below] {$\infty$};
	\draw[very thick,dotted] (0,0)--(3,1.05)--(5.5, 0.7);
	\draw[very thick,dotted] (6, 0.3)--(9, 0.6);
	\node at (4.5,0.6) {$V_{\bad}$};
	\node at (7.5,0.65) {$V_{\bad}$};
	\draw (0,0)--(3,1.05)--(5.5, 1.8);
	\draw (6, 1.4)--(9, 2);
	\node at (4.5, 1.8) {$V_{\good}$};
	\node at (7.5, 2) {$V_{\good}$};
	\draw[<->,>=stealth,dashed] (9, 0.6)--(9,0);
	\node[right] at (9, 0.3) {(b)};
	\draw[<->,>=stealth,dashed] (6, 1.4)--(6,0);
	\node[right] at (6, 0.7) {(a)};
	\draw[dashed] (5.5, 0.3)--(6, 0.3);
	\draw[<->,>=stealth,dashed] (5.5, 0.7)--(5.5, 0.3);
	\node[left] at (5.5, 0.5) {(e)};
	\draw[dashed] (5.5, 1.8)--(6, 1.8);
	\draw[<->,>=stealth,dashed] (6, 1.8)--(6, 1.4);
	\node[left] at (6.1, 1.6) {(c)};
		\draw[dashed] (3, 1.05)--(9, 1.05);
		\draw[<->,>=stealth,dashed] (9, 1.05)--(9, 0.6);
		\node[right] at (9, 0.8) {(d)};
	\end{tikzpicture}
\end{figure}
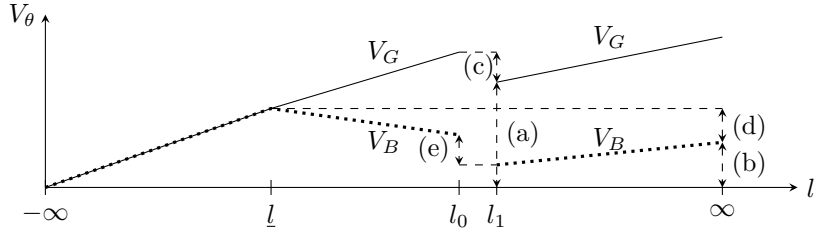

\begin{prop}
	\label{prop:poissonstrangeequil2}
	Fix $\underline{l},l_1\in\mathbb{R}$ with $\underline{l}<l_1$ and fix $\overline{l}\in(l_1,\infty]$. If 
	\begin{enumerate}[(a)]
		\item $\inf\set{V_{\good}(l) -\frac{\beta(j(l))}{r}:l_1\leq l<\overline{l}}\geq \frac{c'_{\good}(1)}{\lambda}$, 
		\item $\sup\set{V_{\bad}(l) -\frac{\beta(j(l))}{r}:l_1\leq l<\overline{l}} \leq \frac{c'_{\bad}(0)}{\lambda}$, 
		\item $\overline{V}^{>}_{\good} -\underline{V}^{1}_{\good}\leq \frac{c'_{\good}(0)}{\lambda}$,
		\item $\overline{V}^{>}_{\bad}-\overline{V}^{1}_{\bad}< \frac{c'_{\bad}(1)}{\lambda}$,
		\item $\underline{V}^{>}_{\bad}-\underline{V}^{1}_{\bad}\geq \frac{c'_{\bad}\left(1-d/\lambda-(\lambda+d)\exp(\underline{l}-l_1 )/\lambda\right)}{\lambda}$,
		\item $\lim_{l\rightarrow \overline{l}}j(l)\leq \underline{l}$,
	\end{enumerate}
	then there exists an equilibrium in which 
	\begin{itemize}
		\item[] $e_{\bad}^*(l)>e_{\good}^*(l)=0$ if $l\in(\underline{l},l_1]$,
		\item[] $e_{\bad}^*(l)=0$, $e_{\good}^*(l)=1$ if $l\in(l_1,\overline{l})$,
		\item[] $e_{\bad}^*(l)=e_{\good}^*(l)=0$ if $l\notin(\underline{l},\overline{l})$.
	\end{itemize}
\end{prop}
The proof is in the appendix. 
The sufficient conditions in Prop.~\ref{prop:poissonstrangeequil2} hold even with the strong single crossing of cost (the marginal cost of $\bad$ is uniformly higher than that of $\good$), which makes the result more surprising.

When $d=0$ in Prop.~\ref{prop:poissonstrangeequil2}, the upper bound $\overline{l}$ of the $e_{\good}^*(l)=1$, $e_{\bad}^*(l)=0$ region must be $\infty$. Otherwise $V_{\bad}(\overline{l})-\frac{\beta_{\min}}{r}$ is large enough to induce $\bad$ to exert effort at $\overline{l}$. If, in the scrutiny region, effort is required forever, then the value of type $\bad$ is lowered, which helps to restore incentives. 

Conditions (a)--(e) in Prop.~\ref{prop:poissonstrangeequil2} have a bound on the marginal benefit of avoiding a jump on the LHS and a bound on the marginal cost per unit of jump frequency on the RHS. The marginal benefit is the value difference between the log likelihood ratios before and after a jump. The marginal cost $c_{\theta}'$ is evaluated at a bound on the effort. The rate of jumps absent effort is $\lambda$, which is also the reduction in jump rate per unit of effort. The inequality and the bound on type $\theta$'s effort in each condition ensure the required effort level. For example, in (a), the lower bound on the marginal benefit of avoiding a jump is larger than the marginal cost to $\good$ at the maximal effort $1$ in the scrutiny region. 
In (e), the lower bound on $e_{\bad}^*$ is not $0$, because together with $e_{\good}^*=0$, this would make jump length zero. Instead, the lower bound $1-\frac{d}{\lambda}-\frac{\lambda+d}{\lambda}\exp(\underline{l}-l_1)$ on $e_{\bad}^*$ ensures that the jumps from the early career region land in the scrutiny region. The maximal jump length necessary for this is $l_1-\underline{l}$.
Condition (f) in Prop.~\ref{prop:poissonstrangeequil2} ensures that jumps end in a region where the value function can be calculated in closed form, which is technically convenient. 

Prop.~\ref{prop:poissonstrangeequil2} is useful for checking incentives after numerically calculating the candidate value functions. 
In the region of $l$ where $e_{\good}^*(l)<e_{\bad}^*(l)$, the value functions cannot in general be found in closed form. Numerical simulations must be used. 
Fig.~\ref{fig:swieff2} displays a numerical example of a switched effort equilibrium. Belief $\mu=\frac{\exp(l)}{1+\exp(l)}$ is used on the $x$-axis instead of $l$, in order to display $V_{\theta}(l)$ at all $l$. The $y$-axis has log scale. 
\begin{figure}[h]
	\centering
	\caption{A switched effort equilibrium with $\beta(l)=\frac{\exp(l)}{\exp(l)+\exp(4.2)}$, $c_{\good}(e)=\frac{e}{10}$, $c_{\bad}(e)=\frac{200e}{201}$, $d=0$, $r=\frac{1}{100}$, $\lambda=2$. Blue dashed line: $V_{\good}$, purple line: $V_{\bad}$, black line: $e_{\bad}^*$.
		\label{fig:swieff2}
	}
	\includegraphics[width=0.6\textwidth]{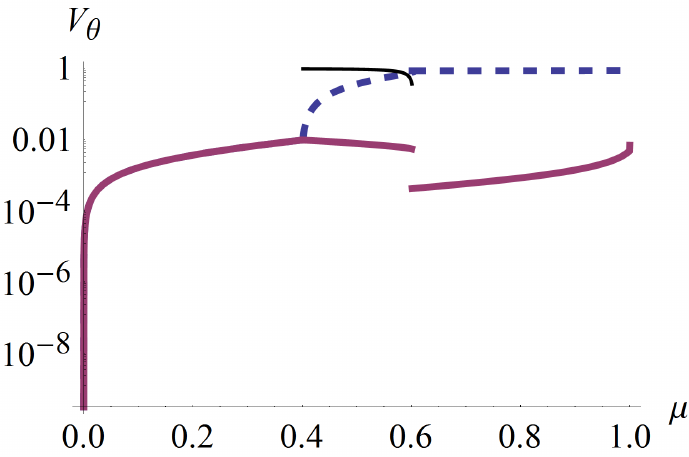}
\end{figure}
The prior probability of $\good$ is $0.6$ at the start of the early career ($l_0=\ln(1.5)$). Insider status is attained when belief $\mu$ has fallen to $0.4$, i.e.\ $\underline{l}=\ln\left(\frac{2}{3}\right)$. 

Relaxing the assumptions on $\beta$ to bounded and weakly increasing, for some $\beta$ the value functions can be found in closed form. Such $\beta$ can be chosen to make the switched effort region unboundedly large (in belief space, from $\epsilon$ to $1-\epsilon$ for any $\epsilon\in(0,\frac{1}{2})$).\footnote{Example available upon request.} 

The next proposition provides sufficient conditions for a switched effort equilibrium in terms of primitives. The value functions are replaced with their closed forms when $e_{\bad}^*,e_{\good}^*\in\set{0,1}$, but bounded by parameters otherwise. The bounding makes the results less tight than in Prop.~\ref{prop:poissonstrangeequil2}.
Parts (a)--(f) in Prop.~\ref{prop:poissonstrangeequil3} and Prop.~\ref{prop:poissonstrangeequil2} correspond, in the sense that each inequality in Prop.~\ref{prop:poissonstrangeequil3} is a bound on the respective inequality in Prop.~\ref{prop:poissonstrangeequil2}. The intuition for Prop.~\ref{prop:poissonstrangeequil3} (and Prop.~\ref{prop:poissonstrangeequil4} to follow) is thus the same as for Prop.~\ref{prop:poissonstrangeequil2}, discussed above and illustrated in Fig.~\ref{fig:incent}. 
\begin{prop}
	\label{prop:poissonstrangeequil3}
	Fix $\underline{l},l_1\in\mathbb{R}$ with $\underline{l}<l_1$. If 
	\begin{enumerate}[(a)]
		\item $\int_{l_1}^{\infty}\frac{\beta(z)-c_{\good}(1)}{\lambda}\exp\left(-r\frac{z-l}{\lambda}\right)dz -\frac{\beta_{\min}}{r} \geq \frac{c'_{\good}(1)}{\lambda}$, 
		\item $\frac{\beta_{\max}}{r+\lambda}+\frac{\lambda\beta_{\min}}{r(r+\lambda)} -\frac{\beta_{\min}}{r} \leq \frac{c'_{\bad}(0)}{\lambda}$, 
		\item $\frac{\beta(\underline{l})}{r} -\int_{l_1}^{\infty}\frac{\beta(z)-c_{\good}(1)}{\lambda}\exp\left(-r\frac{z-l}{\lambda}\right)dz< \frac{c'_{\good}(0)}{\lambda}$,
		\item $\frac{\beta(\underline{l})}{r}-\left[\frac{\beta_{\max}}{r+\lambda}+\frac{\lambda\beta_{\min}}{r(r+\lambda)}\right]< \frac{c'_{\bad}(1)}{\lambda}$,
		\item $\frac{\beta(\underline{l})}{r}-\int_{l_1}^{\infty}\left[\frac{\beta(z)}{\lambda}+\frac{\beta_{\min}}{r}\right]\exp\left(-(r+\lambda)\frac{z-l}{\lambda}\right)dz>  \frac{c_{\bad}'\left(1-\exp(\underline{l}-l_1)\right)}{\lambda} $,
		\item $d=0$,
	\end{enumerate}
	then there exists $l_0\in(\underline{l},l_1)$ and an equilibrium in which 
	\begin{itemize}
		\item[] $e_{\bad}^*(l)>e_{\good}^*(l)=0$ if $l\in(\underline{l},l_0]$,
		\item[] $e_{\bad}^*(l)=0$, $e_{\good}^*(l)=1$ if $l\in[l_1,\infty)$,
		\item[] $e_{\bad}^*(l)=e_{\good}^*(l)=0$ if $l\notin(\underline{l},l_0]\cup[l_1,\infty)$.
	\end{itemize}
\end{prop}
The proof is in the appendix. 

Prop.~\ref{prop:poissonstrangeequil3} requires $d=0$, so there exists an effort level making the signal rate zero. Prop.~\ref{prop:poissonstrangeequil4} in the appendix covers the case $d>0$, so jump length is uniformly bounded. The expressions become more complicated, but the idea is the same as in Prop.~\ref{prop:poissonstrangeequil3}. 

Whenever one switched effort equilibrium exists, there is a continuum of them: if the endpoints of the switched effort region are shifted slightly and if $e_{\bad}^*(l)$ is adjusted to ensure that~(\ref{cornerBR2}) in the appendix holds, then the result is again an equilibrium.
More formally, the sufficient conditions in Prop.~\ref{prop:poissonstrangeequil3} and Prop.~\ref{prop:poissonstrangeequil4} hold in a nonempty open set of parameters. 

If $d>0$, then the log likelihood ratio never reaches $\pm\infty$, so type is never learned with certainty. There is a positive probability that when learning ends, either the $\bad$ type has reached $l>l_0$ or $\good$ has reached $l<l_0$, so a possibility exists of a `wrong' permanent reputation.

\subsection{Other equilibria} 

Equilibria in which $\bad$ exerts minimal effort and $\good$ maximal in some interval, outside which both types exert minimal effort are called \emph{extremal effort equilibria}. Extremal efforts occur in the scrutiny region of switched effort equilibria. Sufficient conditions for the existence of extremal effort equilibria are parts (a),(b) of Prop.~\ref{prop:poissonstrangeequil2} or of the subsequent Propositions. Value functions are given in closed form in the appendix in (\ref{bnewsodesols2}). 

If $d=0$, then the scrutiny region (where $e_{\bad}^*(l)=0$, $e_{\good}^*(l)=1$) may extend from some $l_1\in\mathbb{R}$ to $\infty$, as seen in Prop.~\ref{prop:poissonstrangeequil3} (a),(b). The larger the log likelihood ratio becomes, the greater the incentive to exert effort. This differs from the literature on signalling with Brownian noise, where efforts are highest for intermediate beliefs and go to zero as type becomes certain. The reason is that belief responds less to the Brownian signal when less uncertainty remains about the type. In the current paper, the effort incentive is determined by the jump length. Jumps go from $l$ to $-\infty$ when $d=0$, $e_{\bad}^*(l)=0$ and $e_{\good}^*(l)=1$. 

If $d>0$, then the scrutiny region is bounded above and below, as shown in Lemma~\ref{lem:scrubound}, which is proved in the appendix. 
Efforts are greatest at intermediate beliefs and go to zero as the market becomes certain of the type. 
\begin{lemma}
	\label{lem:scrubound}
If $d>0$, then $\sup\set{l:e_{\bad}^*(l)=0,\;e_{\good}^*(l)=1}<\infty$ and 

\noindent $\inf\set{l:e_{\bad}^*(l)=0,\;e_{\good}^*(l)=1}>-\infty$. 
\end{lemma}

Equilibria in which $e_{\good}^*(l)=e_{\bad}^*(l)>0$ for some $l$ do not exist, because the log likelihood ratio does not respond to signals when $e_{\good}^*(l)=e_{\bad}^*(l)$. Both types will then save cost by minimising effort. This completes the characterisation of equilibria in which $e_{\good}^*(l),e_{\bad}^*(l)\in\set{0,1}$. Such equilibria are either extremal effort or pooling on $0$. Equilibria with interior efforts $e_{\good}^*(l)\in(0,1)$ or $e_{\bad}^*(l)\in(0,1)$ are more difficult to characterise. The results about these (other than switched effort equilibria) are summarised next. 

There are no equilibria in which $e_{\good}^*(l)<e_{\bad}^*(l)$ in a positive-length interval $[l_a,l_b]$, and $e_{\good}^*(l)=e_{\bad}^*(l)=0$ for $l>l_b$. 
The endpoint $l_b$ must be finite, because the incentives require $V_{\bad}$ to decrease over the course of the jump. If jumps up continue indefinitely, then the flow benefit keeps increasing, so $V_{\bad}$ increases. The incentives require the value functions after a jump to be strictly ordered: $V_{\good}(j(l))>V_{\bad}(j(l))$. If the jump ends at $j(l)$ with $e_{\good}^*(j(l))=e_{\bad}^*(j(l))=0$, as occurs from close to the right endpoint $l_b$ of $[l_a,l_b]$, then $V_{\good}(j(l))=V_{\bad}(j(l))=\frac{\beta(l)}{r}$, a contradiction. The scrutiny region of switched effort equilibria is thus indispensable. 

If the cost functions are linear, then there are no equilibria in which $1>e_{\good}^*(l)>e_{\bad}^*(l)\geq0$ in an interval $[l_a,l_b]$ of positive length, with $e_{\good}^*(l)=e_{\bad}^*(l)=0$ elsewhere. Linearity of cost requires indifference between efforts $0,1$ at $l$ for $e_{\theta}^*(l)\in(0,1)$. The value at $l$ is then the same as when the chosen effort $e_{\theta}(l)$ is zero, but the receivers expect the equilibrium efforts $1>e_{\good}^*(l)>e_{\bad}^*(l)\geq0$. This holds at all $l\in[l_a,l_b]$. If the chosen effort is zero everywhere, then the payoffs to the types are identical everywhere. 
Therefore $V_{\good}(l)=V_{\bad}(l)$ at all $l$. In particular, $V_{\good}(j(l))-V_{\good}(l)=V_{\bad}(j(l))-V_{\bad}(l)$. Then the strong single crossing of the cost functions precludes both types from being indifferent between efforts $0,1$ simultaneously. 
The result extends to convex costs that are uniformly close to linear functions, because the model is continuous in the parameters. 

If $d=0$ and the cost functions are linear, then there are no equilibria in which $e_{\bad}^*(l)\in(0,1)$ and $e_{\good}^*(l)=1$ for $l$ in an interval $(l_a,l_b)$ of positive length. The log likelihood ratio would jump to $-\infty$ from $(l_a,l_b)$, so if $\bad$ is indifferent between efforts $0,1$ at $l_c\in(l_a,l_b)$, then $\bad$ strictly prefers $0$ at any $l<l_c$ and $1$ at any $l>l_c$. 


\subsection{Exogenous information revelation}

Switched effort equilibria also exist for some parameter values when the Poisson rate of jumps is $\lambda(1-e_{\theta})+d_{\theta}$ for type $\theta$, where $d_{\good}\neq d_{\bad}$ and $d_{\good},d_{\bad}>0$. In this case, if the efforts of the types are the same, then the signal is informative and $l$ still moves. Pooling equilibria may not exist, because $l$ responding to signals may motivate the sender to exert effort. The possibility of switched efforts with exogenous information revelation is less surprising than in the pure signalling case, because the countersignalling result of \cite{feltovich+2002} uses exogenous info. The mechanism of countersignalling relies on there being at least three types. This paper has two types, so the mechanism for switched efforts is distinct from countersignalling. 

The value functions will be proved jointly continuous in $d_{\good}$ and $d_{\bad}$, so if the assumptions of 
Prop.~\ref{prop:poissonstrangeequil2} hold strictly, then there is a nonempty open set of $d_{\good},d_{\bad}$ for which there exists a switched effort equilibrium. 
\begin{prop}
	\label{prop:dtheta}
	If 
Prop.~\ref{prop:poissonstrangeequil2} (a)--(f) hold strictly, then there exists $\delta>0$ s.t.\ for any $0<d_{\good}<d_{\bad}$ with $|d_{\theta}-d|<\delta$, there exists a switched effort equilibrium. 
\end{prop}
The proof is in the appendix. If $d_{\good}<d_{\bad}$ and $e_{\good}^*(l)=e_{\bad}^*(l)$, then $l$ drifts up and jumps down.

A similar result to Prop.~\ref{prop:dtheta} can be derived with $d_{\good}>d_{\bad}$. This is a less intuitive assumption, because the $\good$ type is disadvantaged in the signal, but advantaged in the cost. If $d_{\good}>d_{\bad}$ and $e_{\good}^*(l)=e_{\bad}^*(l)$, then the drift of $l$ goes down and the jumps up. Switched effort equilibria are still continuous in $|d_{\theta}-d|$. 

\subsection{Robustness}

A continuity argument analogous to Prop.~\ref{prop:dtheta} shows that switched effort equilibria exist when the signal structure is a L\'{e}vy process in which the dominant component is the Poisson process considered in the benchmark model in Section~\ref{sec:setup}.

Adding more types with a two-peaked distribution is expected to yield results similar to the two-type case. Again, this follows from the continuity of the value functions in the prior. The state variable (distribution over types) then becomes multidimensional, which complicates the analysis. 

If the benefit depends to a small extent on the true type, the signal, the actual effort or the effort the market expects from the sender, then the equilibrium is again continuous in the extent of the influence of these factors. 

If the environment is modified so that the benefit depends partly on the effort the market expects, then the model becomes somewhat similar to the career concerns situation considered in \cite{holmstrom1999}. The remaining conceptual difference between the present paper and \cite{holmstrom1999} is that, in the present paper, the sender knows his type. If the sender does not know his type, then all types choose the same effort, so switched efforts are impossible. 

A more radical departure from the current model is to make the benefit depend \emph{only} on the effort the market expects. Suppose the market expects a strategy consisting of switched efforts on $(\underline{l},l_1]$, efforts $e_{\good}^*(l)>e_{\bad}^*(l)$ for $(l_1,\overline{l})$, and pooling elsewhere. This strategy is similar to the one in Prop.~\ref{prop:dtheta}. If the market expects this strategy, then the benefit at any $l\in(\underline{l},\overline{l})$ increases in $d_{\bad} -d_{\good}$, and is zero when $d_{\bad} =d_{\good}$. With a large enough $d_{\bad} -d_{\good}$, it may be possible to construct a switched effort equilibrium according to the pattern used above, but its existence cannot be guaranteed based on the continuity of $V_{\theta}$ alone. This existence question is left for future research.

\section{Noiseless switched efforts}

Equilibria in which $\bad$ exerts more effort than $\good$ for some belief can be constructed in noiseless discrete time signalling using belief threats. This is less surprising than the existence of a switched effort equilibrium with full-support noise, because the noise forces Bayes' rule to apply everywhere. To the author's knowledge, the possibility of switched efforts in pure signalling models (that have no exogenous information revelation) has not been suggested in the literature. Such equilibria may exist in other models, but other authors have not mentioned it.

In this model, time is discrete and the horizon infinite. The types are $\theta=\good,\bad$. The initial log likelihood ratio is $l_0$. Each period, the sender chooses effort $e\in\mathbb{R}_+$, which the receivers observe and use to update their log likelihood ratio. The per-period payoff is $\beta(l)-\hat{c}_{\theta}\cdot e$ for type $\theta$, with $\beta(l)=\frac{\exp(l)}{1+\exp(l)}$ (risk neutral sender) and $\hat{c}_{\bad}>\hat{c}_{\good}>0$. The discount factor is $\delta$. 

In the equilibrium constructed, during the first period, $\bad$ chooses $e_0>0$, while $\good$ puts probability $q_{\good}\in(0,1)$ on $e_0$ and $1-q_{\good}$ on $0$. After $e_0$ is publicly observed in the first period, both types choose $e=0$ forever. 
After $0$ is seen in the first period, $\bad$ exerts zero effort forever, but $\good$ chooses $e_1>0$ forever. The log likelihood ratio does not respond to deviations after $e_0$. After $0$ in the first period, as long as $e_1$ has occurred every period, $l=\infty$. A deviation from $e_1$ after $0$ in the first period is followed by $l=-\infty$ and zero effort forever. This is the belief threat that sustains effort by $\good$ under scrutiny. Any actions not specified above are punished with $l=-\infty$ forever. 

After $0$ is seen in the first period, the constraints for $\good$ to choose $e_1$ and $\bad$ to choose $0$ are 
$\beta_{\max}-\hat{c}_{\good}e_1 \geq\beta_{\min} \geq \beta_{\max}-\hat{c}_{\bad}e_1$. 
The constraint for $\bad$ to exert $e_0$ in the first period is 
\begin{align*}
\frac{1}{1-\delta}\frac{q_{\good}\exp(l_0)}{1+q_{\good}\exp(l_0)}-\hat{c}_{\bad}e_0 \geq \beta_{\max}+\frac{\delta}{1-\delta}\beta_{\min}.
\end{align*}
For $\good$ to be indifferent between $0$ and $e_0$ in the first period, it must be that
\begin{align*}
&\frac{1}{1-\delta}\frac{q_{\good}\exp(l_0)}{1+q_{\good}\exp(l_0)} -\hat{c}_{\good}e_0 
=\beta_{\max}+\frac{\delta}{1-\delta}(\beta_{\max}-\hat{c}_{\good}e_1).
\end{align*}

Take $\delta=\frac{9}{10}$, $c_{\good}=1$, $c_{\bad}=2$, $l_0=0$. The above constraints hold for $e_0=\frac{1}{2}$, $q_{\good}=\frac{1}{3}$ and $e_1=\frac{9}{10}$. This constitutes a switched effort equilibrium. Such equilibria exist for a nonempty open set of parameter values. 

In this model, it is not possible that $\bad$ exerts more effort than $\good$ with probability one, because then high effort would reveal $\bad$, yielding the minimal flow payoff thereafter. Then $\bad$ would deviate to the lowest effort level that the market expects from $\good$. Full separation cannot occur via switched efforts. 

Adding noise makes discrete time repeated signalling intractable, because the state variable (log likelihood ratio) then takes values in an infinite discrete set that is not a regular grid. Neither differential nor difference equations can be used.

\section{Conclusion}

Repeating pure signalling permits equilibria in which the high-cost type exerts more effort than the low-cost for some beliefs of the receivers. These equilibria are numerous and occur both with and without noise, fully revealing signals or exogenous information revelation. To the author's knowledge, the literature has not alluded to the possibility of such equilibria in a pure signalling context. 

Higher signalling effort by the high-cost types can be interpreted as a sign of insecurity---they are trying to avoid future information revelation and effort. Intuitively, the weak act tough to deter attack, because they know they could not handle it. Similarly, the guilty avoid an investigation. The low-cost types know that they can compensate in the future for current low effort, should such compensation be necessary. They also know that future information revelation is likely to be good for them. 

Switched effort equilibria overturn one of the key intuitions from previous signalling models. A good signal can go bad and then become good again, meaning that the same observation may raise belief at one point of the game and lower it at another. Single crossing in type and cost does not carry over to single crossing in type and action. In fact, at some point in the game, effort levels are ordered in the exact reverse order to that commonly found in the literature. This result does not appear in countersignalling, where the high types pool with the low and separate from the medium types. 

Effort and cost are spread more equally across types in a switched effort equilibrium than with standard separation, because the bad type also exerts positive effort. If signalling has a positive externality (outside the scope of the current model, e.g.\ education is good for civil society) and the social benefit from the externality outweighs its cost, then it may be encouraged as a matter of policy. Fairness considerations may then imply a preference for switched efforts, e.g.\ to lead the low-ability workers to acquire at least a minimum of education. Switched effort equilibria spread the cost more evenly than extremal efforts, but less evenly than pooling. However, the zero education in pooling may be undesirable, leaving switched efforts as the best compromise.

\appendix
\section{}
\label{sec:proofs}

\begin{proof}[Proof of Prop.~\ref{prop:poissonstrangeequil2}.]
	Condition (f) implies that if  $e_{\good}^*(l)=1$, $e_{\bad}^*(l)=0$, then $e_{\good}^*(j(l))=e_{\bad}^*(j(l))=0$. This implies $V_{\theta}(j(l))=\frac{\beta(j(l))}{r}$.
	
	The Hamilton-Jacobi-Bellman (HJB) equation for type $\theta$ is
	\begin{align}
	\label{HJB2}
	rV_{\theta}(l)&=\beta(l)+\lambda\left[e_{\good}^*(l)-e_{\bad}^*(l)\right] V_{\theta}'(l) 
	\\&+\max_e\set{[\lambda(1-e)+d]\left[V_{\theta}\left(j(l)\right)-V_{\theta}(l)\right]-c_{\theta}(e)}. \notag
	\end{align}
	The best response solves $\lambda\left[V_{\theta}(l) -V_{\theta}\left(j(l)\right)\right]=c'_{\theta}(e)$ if interior. Corner solutions for the best response are given by 
	\begin{align}
	\label{cornerBR2}
	e_{\theta}(l)=\begin{cases}
	0 & \text{ if } \lambda\left[V_{\theta}(l) -V_{\theta}\left(j(l)\right)\right]\leq c'_{\theta}(0), \\
	1 & \text{ if } \lambda\left[V_{\theta}(l) -V_{\theta}\left(j(l)\right)\right]\geq c'_{\theta}(1). \\
	\end{cases}
	\end{align}
	A verification theorem (Theorem~4.6 in \cite{presman+1990} as modified for the discounted case in~\cite{yushkevich1988}) is used to check that the solutions of~(\ref{HJB2}) coincide with the value functions. 
	
	In the   $e_{\good}^*(l)=1$, $e_{\bad}^*(l)=0$ region, condition~(a) implies $V_{\good}(l) -V_{\good}\left(j(l)\right)\geq \frac{c'_{\good}(1)}{\lambda}$ for all $l$, because $V_{\good}(l)\geq \underline{V}^1_{\good}$ and $V_{\good}(j(l))=\frac{\beta(j(l))}{r}$. In turn, $V_{\good}(l) -V_{\good}\left(j(l)\right)\geq \frac{c'_{\good}(1)}{\lambda}$ implies $e_{\good}(l)=1$ based on~(\ref{cornerBR2}). Therefore condition~(a) suffices for $e_{\good}(l)=1\;\forall l$ in the region with   $e_{\good}^*(l)=1$, $e_{\bad}^*(l)=0$. 
	
	Condition~(b) implies $V_{\bad}(l) -V_{\bad}\left(j(l)\right)\leq \frac{c'_{\bad}(0)}{\lambda}$, because $V_{\bad}(l)\leq \overline{V}^1_{\bad}$ and $V_{\bad}(j(l))=\frac{\beta(j(l))}{r}$. In turn, $V_{\bad}(l) -V_{\bad}\left(j(l)\right)\leq \frac{c'_{\bad}(0)}{\lambda}$ implies $e_{\bad}(l)=0$ based on~(\ref{cornerBR2}). Therefore condition~(b) suffices for $e_{\bad}(l)=0\;\forall l$ in the region with   $e_{\good}^*(l)=1$, $e_{\bad}^*(l)=0$. 
	
	If $e_{\bad}^*(l)>e_{\good}^*(l)=0$, then $\lim_{e_{\bad}^*(l)\rightarrow 1}j(l)=l+\lim_{e_{\bad}^*(l)\rightarrow 1}\ln\frac{1+d}{1-e_{\bad}^*(l)+d}\geq \overline{l}$ by (\ref{jumpb}) and (f). 
	In the equilibrium constructed, $e_{\bad}^*(l)$ is chosen in the $e_{\bad}^*(l)>e_{\good}^*(l)=0$ region to ensure the jumps end in the   $e_{\good}^*(l)=1$, $e_{\bad}^*(l)=0$ region (formally $e_{\good}^*(j(l))=1,e_{\bad}^*(j(l))=0$). 
	Condition~(c) implies $V_{\good}(l) -V_{\good}\left(j(l)\right)\leq \frac{c'_{\good}(0)}{\lambda}$, which in turn implies $e_{\good}(l)=0$ based on~(\ref{cornerBR2}). Therefore condition~(c) suffices for $e_{\good}(l)=0\;\forall l$ in the region with $e_{\bad}^*(l)>e_{\good}^*(l)=0$. 
	
	To show that~(d) and~(e) imply the existence of an equilibrium level of $e_{\bad}^*(l)$ at each $l$ in the $e_{\bad}^*(l)>e_{\good}^*(l)=0$ region, first the candidate value functions in the   $e_{\good}^*(l)=1$, $e_{\bad}^*(l)=0$ region are calculated. Substituting $e_{\bad}^*(l)=e_{\bad}(l)=0$ and $e_{\good}^*(l)=e_{\good}(l)=1$ into each type's HJB equation and solving the resulting ordinary differential equations (ODEs) yields
	\begin{align}
	\label{bnewsodesols2}
	V_{\good}(l)&=\exp\left(-\frac{(r+d)(\overline{l}-l)}{\lambda}\right)V_{\good}(\overline{l}) \notag
	\\&+ \int_{l}^{\overline{l}}\left[\frac{\beta(z)-c_{\good}(1)}{\lambda}+\frac{\beta(j(z))d}{r\lambda}\right]\exp\left(-\frac{(r+d)(z-l)}{\lambda}\right)dz,
	\\ V_{\bad}(l)&=\exp\left(-(r+\lambda+d)\frac{\overline{l}-l}{\lambda}\right)V_{\bad}(\overline{l}) \notag
	\\& +\int_{l}^{\overline{l}}\left[\frac{\beta(z)}{\lambda}+\frac{(\lambda+d)\beta(j(z))}{r\lambda}\right]\exp\left(-(r+\lambda+d)\frac{z-l}{\lambda}\right)dz,\notag
	\end{align}
	where $j(z)=z+\ln\left(\frac{d}{\lambda+d}\right)$. 
	If $\overline{l}$ is finite, then value matching gives $\lim_{l\rightarrow \overline{l}}V_{\theta}(l)=V_{\theta}(\overline{l})=\frac{\beta(\overline{l})}{r}$, which provides the boundary condition for the ODEs in~(\ref{bnewsodesols2}). 
	By~(\ref{bnewsodesols2}), $V_{\theta}$ is strictly increasing in $l$, so $\underline{V}^1_{\theta}=V_{\theta}(l_1)$ and $\overline{V}^1_{\theta}=\lim_{l\rightarrow \overline{l}}V_{\theta}(l)$. The limit is relevant when $\overline{l}=\infty$, in which case $\lim_{l\rightarrow \overline{l}}V_{\theta}(l)< \frac{\beta_{\max}}{r}$, as can be seen from~(\ref{bnewsodesols2}). 
	
	
	Given conditions (a)--(c) and (f), $e_{\bad}^*(l)$ in the region with $e_{\bad}^*(l)>e_{\good}^*(l)=0$ is part of an equilibrium iff 
	\begin{align}
	\label{equilcond2}
	\lambda\left[V_{\theta}(l) -V_{\theta}\left(l+\frac{\lambda+d}{\lambda(1-e_{\bad}^*(l))+d}\right)\right]=c'_{\theta}(e_{\bad}^*(l)).
	\end{align} 
	This condition merely says $e_{\bad}^*(l)$ is a best response to $e_{\good}^*(l)=0$ and itself. Note that the left hand side (LHS) of~(\ref{equilcond2}) strictly decreases in $e_{\bad}^*(l)$ and the right hand side (RHS) increases. The LHS is continuous, because $l$ is fixed and $V_{\theta}\left(l+\frac{\lambda+d}{\lambda(1-e_{\bad}^*(l))+d}\right)$ is continuous by~(\ref{bnewsodesols2}). The RHS is continuous by assumption. 
	
	Condition~(d) implies $V_{\bad}(l)-\overline{V}^{1}_{\bad}< \frac{c'_{\bad}(1)}{\lambda}$ and (e) implies $V_{\bad}(l)-\underline{V}^{1}_{\bad}\geq \frac{c'_{\bad}(\hat{e})}{\lambda}$ for $\hat{e}$ solving $\underline{l}+\ln\left(\frac{\lambda+d}{\lambda(1-\hat{e})+d}\right)=l_1$. These bounds, the continuity of~(\ref{equilcond2}) and the Mean Value Theorem imply that there exists $e_{\bad}^*(l)\in[\hat{e},1)$ s.t.~(\ref{equilcond2}) holds. 
\end{proof}

\begin{proof}[Proof of Prop.~\ref{prop:poissonstrangeequil3}.]
	(a) $\underline{V}^1_{\good}=\int_{l_1}^{\infty}\frac{\beta(z)-c_{\good}(1)}{\lambda}\exp\left(-r\frac{z-l}{\lambda}\right)dz$ when $\overline{l}=\infty$ based on~(\ref{bnewsodesols2}). When  $e_{\good}^*(l)=1$, $e_{\bad}^*(l)=0$ and $d=0$, then $j(l)=l+\ln\frac{0}{1}=-\infty$ by~(\ref{jumpb}), so $V_{\theta}(j(l))=\frac{\beta_{\min}}{r}$. By~(\ref{cornerBR2}), condition~(a) suffices for $e_{\good}(l)=1$ for any $l$ s.t.\   $e_{\good}^*(l)=1$, $e_{\bad}^*(l)=0$.
	
	(b) $\lim_{l\rightarrow\infty}V_{\bad}(l)=\frac{\beta_{\max}}{r+\lambda}+\frac{\lambda\beta_{\min}}{r(r+\lambda)}$ when $\overline{l}=\infty$ based on~(\ref{bnewsodesols2}). By~(\ref{cornerBR2}), condition~(b) suffices for $e_{\bad}(l)=0$ for any $l$ s.t.\ $e_{\good}^*(l)=1$, $e_{\bad}^*(l)=0$, because $j(l)=-\infty$. Equilibrium behaviour in the   $e_{\good}^*(l)=1$, $e_{\bad}^*(l)=0$ region is thus guaranteed. 
	
	(c)--(e). 
	In the $e_{\bad}^*(l)>e_{\good}^*(l)=0$ region, in the absence of signals, $l$ drifts down by~(\ref{poissonbnewsloglr}). At $\hat{l}\in(\underline{l},l_0]$, the probability of reaching $\underline{l}$ is $\exp\left(\int_{\underline{l}}^{\hat{l}}\frac{\lambda(1-e_{\bad}^*(z))}{-\lambda e_{\bad}^*(z)}dz\right)$ for $\bad$, because at $l$, signals occur at rate $\lambda(1-e_{\bad}^*(l))$ and $l$ drifts $\lambda[e_{\good}^*(l)-e_{\bad}^*(l)]$ per unit of time. For $\good$, the probability is 
	$\exp\left(\int_{\underline{l}}^{\hat{l}}\frac{\lambda}{-\lambda e_{\bad}^*(z)}dz\right)$, smaller than for $\bad$ because of lower effort of avoiding jumps. The payoff conditional on reaching $\underline{l}$ is $\frac{\beta(\underline{l})}{r}$ and conditional on not reaching, bounded by $\frac{\beta_{\min}}{r}$ and $\frac{\beta_{\max}}{r}$. The discount rate is $r$ and the time it takes to drift from $l$ to $\underline{l}$ is $\left|\int_{\underline{l}}^{l}\frac{dz}{\lambda[e_{\good}^*(z)-e_{\bad}^*(z)]}\right|$. Therefore for any $l\in(\underline{l},l_0]$, 
	\begin{align}
	\label{Vbounds}
	&\exp\left(\int_{\underline{l}}^{l}\frac{- dz}{ e_{\bad}^*(z)}\right)\exp\left(\int_{\underline{l}}^{l}\frac{-rdz}{\lambda e_{\bad}^*(z)}\right)\frac{\beta(\underline{l})}{r} +\left[1-\exp\left(\int_{\underline{l}}^{l}\frac{\lambda dz}{-\lambda e_{\bad}^*(z)}\right)\right]\frac{\beta_{\max}}{r} \notag
	\\&\geq 
	V_{\theta}(l)\geq 
	\\&\exp\left(\int_{\underline{l}}^{l}\frac{- dz}{ e_{\bad}^*(z)}\right)\exp\left(\int_{\underline{l}}^{l}\frac{-rdz}{\lambda e_{\bad}^*(z)}\right)\frac{\beta(\underline{l})}{r} +\left[1-\exp\left(\int_{\underline{l}}^{l}\frac{\lambda dz}{-\lambda e_{\bad}^*(z)}\right)\right]\frac{\beta_{\min}}{r}. \notag
	\end{align}
	To construct the $e_{\bad}^*(l)>e_{\good}^*(l)=0$ part of the equilibrium, first conjecture $e_{\bad}^*(l)\geq 1-\exp\left(l-l_1\right)\;\forall l\in(\underline{l},l_1)$ so that $e_{\good}^*(j(l))=1$, $e_{\bad}^*(j(l))=0$, second show that $V_{\theta}$ is continuous at $\underline{l}$, third use (c) to ensure $e_{\good}^*(l)=0$, and fourth use (d),(e) to ensure $e_{\bad}^*(l)\in[1-\exp\left(l-l_1\right), 1)$, verifying the conjecture. 
	
	If $e_{\bad}^*(l)\geq \epsilon>0\;\forall l\in(\underline{l},l_0]$, then by~(\ref{Vbounds}) and the Squeeze Theorem, $\lim_{l\rightarrow \underline{l}+}V_{\theta}(l)=\frac{\beta(\underline{l})}{r}$. 
	This and (c) imply the existence of $l_2\in(\underline{l},l_1)$ s.t.\ for all $l\in(\underline{l},l_2]$, $V_{\good}(l)-V_{\good}(j(l))\leq \frac{c'_{\good}(0)}{\lambda}$, which by~(\ref{cornerBR2}) implies $e_{\good}(l)=0$. 
	
	Condition~(d), $\lim_{l\rightarrow\infty}V_{\bad}(l)=\frac{\beta_{\max}}{r+\lambda}+\frac{\lambda\beta_{\min}}{r(r+\lambda)}$ and $\lim_{l\rightarrow \underline{l}+}V_{\theta}(l)=\frac{\beta(\underline{l})}{r}$ imply the existence of $l_3\in(\underline{l},l_1)$ s.t.\ for all $l\in(\underline{l},l_3]$, $V_{\bad}(l)-V_{\bad}(j(l))< \frac{c'_{\bad}(1)}{\lambda}$, which by~(\ref{cornerBR2}) implies $e_{\bad}(l)<1 $. 
	
	Condition~(e) similarly implies the existence of $l_4\in(\underline{l},l_1)$ s.t.\ for all $l\in(\underline{l},l_4]$, $V_{\bad}(l)-V_{\bad}(j(l))\geq \frac{c'_{\bad}\left(1-\exp(l -l_1)\right)}{\lambda}$, which by~(\ref{cornerBR2}) implies $e_{\bad}(l)\geq 1-\exp\left(l-l_1\right)$. 
	
	Set $l_0:=\min\set{l_2,l_3,l_4}$. As in the proof of Prop.~\ref{prop:poissonstrangeequil2}, the bounds from conditions (d),(e) and the Mean Value Theorem imply that there exists $e_{\bad}^*(l)\in[1-\exp\left(l-l_1\right), 1)$ s.t.~(\ref{equilcond2}) holds. 
	
\end{proof}

\begin{prop}
	\label{prop:poissonstrangeequil4}
	Fix $\underline{l},l_1,\overline{l}\in\mathbb{R}$ with $\underline{l}<l_1<\overline{l}$. If $d>0$ and
	\begin{enumerate}[(a)]
		\item $\min_{l\in[l_1,\overline{l}]}\set{V_{\good}(l) -\frac{\beta\left(l+\ln\left(\frac{d}{\lambda+d}\right)\right)}{r} } \geq \frac{c'_{\good}(1)}{\lambda}$, 
		\item  $\max_{l\in[l_1,\overline{l}]}\set{V_{\bad}(l) -\frac{\beta\left(l+\ln\left(\frac{d}{\lambda+d}\right)\right)}{r} } \leq \frac{c'_{\bad}(0)}{\lambda}$, 
		\item $\frac{\beta(\underline{l})}{r} -V_{\good}(l_1)< \frac{c'_{\good}(0)}{\lambda}$,
		\item $\frac{\beta(\underline{l})}{r} -\frac{\beta(\overline{l})}{r}< \frac{c'_{\bad}(1)}{\lambda}$,
		\item $\frac{\beta(\underline{l})}{r} -V_{\bad}(l_1)> \frac{c'_{\bad}\left(1-d/\lambda-(\lambda+d)\exp(\underline{l}-l_1 )/\lambda\right)}{\lambda}$,
		\item $\overline{l}+\ln\left(\frac{d}{\lambda+d}\right)< \underline{l}$, 
	\end{enumerate}
where $V_{\theta}$ is given in~(\ref{bnewsodesols2}), then there exists $l_0\in(\underline{l},l_1)$ and an equilibrium in which 
	\begin{itemize}
		\item[] $e_{\bad}^*(l)>e_{\good}^*(l)=0$ if $l\in(\underline{l},l_0]$,
		\item[] $e_{\bad}^*(l)=0$, $e_{\good}^*(l)=1$ if $l\in[l_1,\overline{l}]$,
		\item[] $e_{\bad}^*(l)=e_{\good}^*(l)=0$ if $l\notin(\underline{l},l_0]\cup[l_1,\overline{l}]$.
	\end{itemize}
\end{prop}
The proof is omitted, because it is similar to that of Prop.~\ref{prop:poissonstrangeequil3}: value functions or bounds on them are substituted into the conditions of Prop.~\ref{prop:poissonstrangeequil2} to check incentives. 
Condition (d) in Prop.~\ref{prop:poissonstrangeequil4} always holds, but is added for better comparability to the other propositions. 

\begin{proof}[Proof of Prop.~\ref{prop:dtheta}.]
	The equilibrium with jump rate $\lambda(1-e_{\theta})+d_{\theta}$ will be constructed using the same $\underline{l},l_1,\overline{l}$ as in Prop.~\ref{prop:poissonstrangeequil4}. 
	If the Poisson rate is $\lambda(1-e_{\theta})+d_{\theta}$ for type $\theta$, then~(\ref{jumpb}) becomes
	\begin{align}
	\label{jumpdtheta}
	j(l)=l+\ln\left(\frac{\lambda(1-e_{\good}^*(l))+d_{\good}}{\lambda(1-e_{\bad}^*(l))+d_{\bad}}\right).
	\end{align}
	In the absence of jumps, the log likelihood ratio drift is $\lambda[e_{\good}^*(l)-e_{\bad}^*(l)]-d_{\good}+d_{\bad}$. 
	Jumps reach $l_1$ from $l$ in the $e_{\bad}^*(l)>e_{\good}^*(l)$ region iff $e_{\bad}^*(l)\geq 1-\frac{d_{\bad}}{\lambda}-\frac{\lambda+d_{\good}}{\lambda}\exp(l-l_1 )$, which is implied by Prop.~\ref{prop:poissonstrangeequil4}~(f) holding strictly and by $|d_{\theta}-d|<\delta$. 
	
	If $d_{\good}<d_{\bad}$, then $\underline{l}$ is a \emph{stasis point}: to the left of $\underline{l}$ and at $\underline{l}$, the drift of the log likelihood ratio is positive and to the right of $\underline{l}$, negative. After reaching $\underline{l}$, the $l$ process oscillates around it with infinite frequency and zero amplitude, spending fraction $w:=\frac{d_{\good}-d_{\bad}+\lambda \lim_{l\rightarrow \underline{l}+}e_{\bad}^*(l)}{\lambda \lim_{l\rightarrow \underline{l}+}e_{\bad}^*(l)}$ of any time interval infinitesimally to the left of $\underline{l}$ and $1-w$ infinitesimally to the right. 
	The drift at $\underline{l}$ is a mixture of drifts to the left and to the right of it, with weight $w$ on the left-side drift. A fraction $w$ of the time, jumps to $j_{-}(\underline{l}):=\underline{l}+\ln\left(\frac{\lambda+d_{\good}}{\lambda+d_{\bad}}\right)$ occur for type $\theta$ at rate $\lambda+d_{\theta}$. A fraction $1-w$ of the time, jumps to $j_{+}(\underline{l}):=\underline{l}+\ln\left(\frac{\lambda+d_{\good}}{\lambda(1-\lim_{l\rightarrow \underline{l}+}e_{\bad}^*(l))+d_{\bad}}\right)>l_1$ occur for $\good$ at rate $\lambda+d_{\theta}$ and for $\bad$ at rate $\lambda_{\bad}^*:=\lambda(1-\lim_{l\rightarrow \underline{l}+}e_{\bad}(l))+d_{\bad}$. 
	The values of the types at $\underline{l}$ are 
	\begin{align}
	\label{Vtypedep}
	&V_{\good}(\underline{l})= \frac{\beta(\underline{l})}{r+\lambda+d_{\good}} +\frac{\lambda+d_{\good}}{r+\lambda+d_{\good}}\left[wV_{\good}\left(j_{-}(\underline{l})\right) +(1-w)V_{\good}\left(j_{+}(\underline{l})\right)\right], \notag
	\\& V_{\bad}(\underline{l})= \frac{\beta(\underline{l})}{(\lambda+d_{\bad})w +\lambda_{\bad}^*(1-w)} +\frac{(\lambda+d_{\bad})w}{(\lambda+d_{\bad})w +\lambda_{\bad}^*(1-w)} V_{\bad}\left(j_{-}(\underline{l})\right)  \\&+\frac{\lambda_{\bad}^*(1-w)}{(\lambda+d_{\bad})w +\lambda_{\bad}^*(1-w)}V_{\bad}\left(j_{+}(\underline{l})\right). \notag
	\end{align}
	Value is continuous at $\underline{l}$ by the same argument as in the proof of Prop.~\ref{prop:poissonstrangeequil3}. 
	For any $\epsilon>0$ there exists $\delta>0$ s.t.\ if $|d_{\theta}-d|<\delta$, then $w>1-\epsilon$ and $|\underline{l}-j_{-}(\underline{l})|<\epsilon$. Then by~(\ref{Vtypedep}) and the continuity of $V_{\theta}$,
	\begin{align*}
	&\frac{\beta(\underline{l})}{r+\lambda+d_{\good}} +\frac{\lambda+d_{\good}}{r+\lambda+d_{\good}}\epsilon \frac{\beta_{\max}}{r}
	\geq V_{\good}(\underline{l})-\frac{(\lambda+d_{\good})(1-\epsilon)}{r+\lambda+d_{\good}}[V_{\good}\left(j_{-}(\underline{l})\right)+\epsilon],
	\\&V_{\good}(\underline{l})-\frac{(\lambda+d_{\good})(1-\epsilon)}{r+\lambda+d_{\good}}[V_{\good}\left(j_{-}(\underline{l})\right)-\epsilon]\geq  \frac{\beta(\underline{l})}{r+\lambda+d_{\good}} +\frac{\lambda+d_{\good}}{r+\lambda+d_{\good}}\epsilon \frac{\beta_{\min}}{r}, 
	\end{align*}
	so by the Squeeze Theorem, $\lim_{d_{\good},d_{\bad}\rightarrow d}V_{\good}(\underline{l})=\frac{\beta(\underline{l})}{r}$. 
	A similar argument shows $\lim_{d_{\good},d_{\bad}\rightarrow d}V_{\bad}(\underline{l})=\frac{\beta(\underline{l})}{r}$. 
	
	At any $l$ s.t.\ $e_{\bad}^*(l)=e_{\good}^*(l)=0$, for any $T,\epsilon>0$ there exists $\delta>0$ s.t.\ if $|d_{\theta}-d|<\delta$, then with probability $1-\epsilon$, the flow payoff remains within $\epsilon$ of $\beta(l)$ for at least $T$ units of time. This is because by~(\ref{jumpdtheta}), $\lim_{d_{\theta}\rightarrow d}j(l)=l$, and when $e_{\bad}^*(l)=e_{\good}^*(l)$, the drift of the log likelihood ratio is $-d_{\good}+d_{\bad}$. Therefore $\lim_{d_{\good},d_{\bad}\rightarrow d}V_{\theta}(l)=\frac{\beta(l)}{r}$. 
	
	Using (\ref{jumpdtheta}) and the drift $\lambda[e_{\good}^*(l)-e_{\bad}^*(l)]-d_{\good}+d_{\bad}$ in the HJB equation~(\ref{HJB2}) and solving for $V_{\theta}$ in the $e_{\bad}^*(l)=0$, $e_{\good}^*(l)=1$ region, (\ref{bnewsodesols2}) becomes
	\begin{align}
	\label{dthetaodesols}
	V_{\good}(l)&=\exp\left(-\frac{(r+d_{\good})(\overline{l}-l)}{\lambda -d_{\good}+d_{\bad}}\right)V_{\good}(\overline{l}) \notag
	\\&+ \int_{l}^{\overline{l}}\frac{\beta(z)-c_{\good}(1)+d_{\good}V_{\good}(j(z))}{\lambda -d_{\good}+d_{\bad}}\exp\left(-\frac{(r+d_{\good})(z-l)}{\lambda -d_{\good}+d_{\bad}}\right)dz,
	\\ V_{\bad}(l)&=\exp\left(-\frac{(r+\lambda+d_{\bad})(\overline{l}-l)}{\lambda -d_{\good}+d_{\bad}}\right)V_{\bad}(\overline{l}) \notag
	\\& +\int_{l}^{\overline{l}}\frac{\beta(z)+(d_{\bad}+\lambda)V_{\bad}(j(z))}{\lambda -d_{\good}+d_{\bad}}\exp\left(-\frac{(r+\lambda+d_{\bad})(z-l)}{\lambda -d_{\good}+d_{\bad}}\right)dz.\notag
	\end{align}
	From $e_{\bad}^*(\overline{l})=e_{\good}^*(\overline{l})=0$, it follows that $\lim_{d_{\good},d_{\bad}\rightarrow d}V_{\theta}(\overline{l})=\frac{\beta(\overline{l})}{r}$ in (\ref{dthetaodesols}). From $e_{\bad}^*(j(l))=e_{\good}^*(j(l))=0$, it follows that $\lim_{d_{\good},d_{\bad}\rightarrow d}V_{\theta}(j(l))=\frac{\beta(j(l))}{r}$ in (\ref{dthetaodesols}). Based on these, the limit of (\ref{dthetaodesols}) as $d_{\theta}\rightarrow d$ is (\ref{bnewsodesols2}). 
	
	Since Prop.~\ref{prop:poissonstrangeequil4} (a)--(f) hold strictly and $V_{\theta}(l)$ is continuous in $d_{\good},d_{\bad}$ for any $l$, there exist $l_0\in(\underline{l},l_1)$ and $\delta>0$ s.t.\ if $|d_{\theta}-d|<\delta$, then Prop.~\ref{prop:poissonstrangeequil2} (a)--(f) are satisfied. This suffices for the existence of a switched effort equilibrium. 
\end{proof}

\begin{proof}[Proof of Lemma~\ref{lem:scrubound}.]
	If $\set{l:e_{\bad}^*(l)=0,\;e_{\good}^*(l)=1}=\emptyset$, then the result is vacuously true. Suppose there exists $l$ s.t.\ $e_{\bad}^*(l)=0$, $e_{\good}^*(l)=1$.
	Define $j^n(l):=j(j^{n-1}(l))$ and $j^1(l):=j(l)=l+\ln\left(\frac{d}{\lambda+d}\right)$. Clearly $j(l)-l\in(-\infty,0)$. 
	
	Define $n^*:=\frac{\lambda(\beta_{\max}-\beta_{\min})}{rc'_{\good}(1)}$.  By~(\ref{cornerBR2}), $e_{\good}^*(l)=1$ requires 
	$\lambda\left[V_{\theta}(l) -V_{\theta}\left(j(l)\right)\right]\geq c'_{\theta}(1).$ Since $V_{\theta}(l)\in\left[\frac{\beta_{\min}}{r},\frac{\beta_{\max}}{r}\right]$, we have 
	\begin{align}
	\label{boundedscru}
	&\sup\set{n\in\mathbb{N}:\exists l\text{ s.t. } l,j^1(l),\ldots,j^n(l)\in \set{\hat{l}:e_{\bad}^*(\hat{l})=0,\;e_{\good}^*(\hat{l})=1}}
	\leq n^*.
	\end{align} 
	Pick $l_1<\overline{l}$ s.t.\ 
	for any $\hat{l}\in(l_1,\overline{l})$, $e_{\bad}^*(\hat{l})=0$, $e_{\good}^*(\hat{l})=1$. By~(\ref{boundedscru}), $\overline{l}-l_1\leq n^*/\ln\left(\frac{d}{\lambda+d}\right)<\infty$. 
	Raise $\overline{l}$ maximally, so for any $\eta>0$, $e_{\bad}^*(\overline{l}+\eta)=e_{\good}^*(\overline{l}+\eta)=0$. In this case, by continuity of the $V_{\theta}$ given in~(\ref{bnewsodesols2}), $\lim_{l\rightarrow \overline{l}}V_{\theta}(l)=\frac{\beta(\overline{l})}{r}$. 
	
	By definition of $n^*$, if $e_{\bad}^*(l)=0$ and $e_{\good}^*(l)=1$, then there exists $\hat{n}\in\mathbb{N}$, $\hat{n}<n^*+1$, s.t.\ $e_{\bad}^*(j^{\hat{n}}(l))=e_{\good}^*(j^{\hat{n}}(l))=0$. Let $\underline{n}$ be the minimal such $\hat{n}$ and let $l^*=\overline{l}-\nu$ for $\nu>0$ small enough. 
	
	Due to the bounded $\beta$ and $\beta'>0$, there exist $N,\epsilon>0$ s.t.\ if $|l|>N$, then $\beta(l)-\beta\left(j^{\underline{n}}(l)\right)<r\epsilon$. Take $\epsilon=\frac{c_{\good}'(1)}{\lambda}$. Then there exists $N>0$ s.t.\ if $\overline{l}>N$, then $V_{\theta}(l^*)-V_{\theta}(j^{\underline{n}}(l^*))<\frac{c_{\good}'(1)}{\lambda}$. Since $V_{\theta}$ is strictly increasing in the $e_{\bad}^*(l)=0$, $e_{\good}^*(l)=1$ region, this implies the failure of~(\ref{cornerBR2}) at all $l=l^*,j^1(l^*),\ldots,j^{\underline{n}-1}(l^*)$. Thus there is no incentive for $\good$ to exert maximal effort. This shows $\overline{l}\leq N$ for all intervals $(l_1,\overline{l})$ on which switched efforts occur. 
	
	To show $l_1\geq-N$, note that $V_{\theta}(l)<V_{\theta}(\overline{l})=\frac{\beta(\overline{l})}{r}$ for any $l\in(l_1,\overline{l})$. If $V_{\theta}(\overline{l})-\frac{\beta_{\min}}{r}<\frac{c_{\good}'(1)}{\lambda}$, then~(\ref{cornerBR2}) fails at any $l\in(l_1,\overline{l})$. It was shown above that $\overline{l}-l_1\leq n^*/\ln\left(\frac{d}{\lambda+d}\right)$, so for $\good$ to have an incentive to exert maximal effort, it is necessary that $l_1\geq \beta^{-1}\left(\frac{rc_{\good}'(1)}{\lambda}+\beta_{\min}\right)-n^*/\ln\left(\frac{d}{\lambda+d}\right)$.
\end{proof}

\bibliographystyle{ecta}
\bibliography{teooriaPaberid} 
\end{document}